\documentclass[a4paper,12pt]{article}
\title{A New Look at the Arcsine Law and ``Quantum-Classical Correspondence''}
\author{Hayato Saigo\footnote{E-mail: h\_saigoh@nagahama-i-bio.ac.jp } 
\\ Nagahama Institute of Bio-Science and Technology \\  Nagahama 526-0829, Japan
}
\date{}

\usepackage{amsmath}
\usepackage{amsthm}
\usepackage{amssymb}
\usepackage{amsbsy}
\usepackage{amsfonts}
\usepackage{amstext}
\usepackage{amscd}
\usepackage[dvips]{epsfig}

\numberwithin{equation}{section}
\theoremstyle{plain}
\newtheorem{thm}{Theorem}[section]

\theoremstyle{definition}

\newtheorem{rem}[thm]{Remark}
\newtheorem{df}[thm]{Definition}
\newtheorem{nota}[thm]{Notation}

\begin{document}
\maketitle

\begin{abstract}
We prove 
that the arcsine law as the time-averaged distribution for classical harmonic oscillators 
emerges from the distributions for quantum harmonic oscillators in terms of noncommutative algebraic probability.
This is nothing but a simple and rigorous realization of 
``Quantum-Classical Correspondence'' for harmonic oscillators. 
\end{abstract}

\section{Introduction}
The normalized arcsine law $\mu_{As}$ is the probability distribution on $\mathbb{R}$ with support $[-\sqrt{2}, \sqrt{2}]$ defined as 
\[
\mu_{As} (dx)=\frac{1}{\pi} \frac{dx}{\sqrt{2-x^2}},
\]
whose $n$-th moment $M_n:=\int_{\mathbb{R}}x^n \mu_{As}(dx)$ is given by
\[
M_{2m+1}=0,\:\:\:M_{2m}=\frac{1}{2^m}\binom{2m}{m}.
\]
In this case the moment problem is determinate, that is, the moment sequence $\{M_n\}$ characterizes $\mu_{As}$.
 
The distribution $\mu_{As}$ often appears in classical
 probability theory. In the noncommutative context, it is 
 also known as the limit distribution for ``monotone central limit theorem''(\cite{Mur}, a simple 
 proof is found in \cite{Saigo}). 
Here we discuss another aspect of this distribution: 
the relationship with the classical harmonic oscillator. 

Let $x(t)=A\sin t$ be a classical harmonic oscillator with the amplitude $A$. Then it is easy to see that 
the time-averaged distribution $\mu $ of position $x$  has the form
\[
\mu(dx)=C\frac{dx}{\sqrt{A^2-x^2}}
\]
where $C$ denotes the normalizing constant. In $A=\sqrt{2}$ case, $\mu =\mu_{As}$.

Then a question arises: Is it possible to see whether and in what meaning the ``Quantum-Classical Correspondence'' 
holds for harmonic oscillators? This question, 
which is related to fundamental problems in Quantum theory and asymptotic analysis \cite{EZA}, 
is analyzed and generalized from the viewpoint of noncommutative algebraic probability 
with quite a simple combinatorial argument. 

\section{Basic notions}

Let $\mathcal{A}$ be a $\ast$-algebra. We call a linear map $\varphi :\mathcal{A}\rightarrow \mathbb{C}$ a state on $\mathcal{A}$ if it satisfies
\[
\varphi(1)=1,\:\:\:\varphi (a^{\ast}a)\geq 0.
\]
A pair $(\mathcal{A}, \varphi)$ of a $\ast$-algebra and a state on it is called an algebraic probability space.
Here we adopt a notation for a state $\varphi :\mathcal{A}\rightarrow \mathbb{C}$, an element $X\in \mathcal{A}$ and a probability distribution $\mu$ on $\mathbb{R}$.
\begin{nota} We use the notation 
$X\sim_{\varphi} \mu $
when $\varphi(X^m)=\int_{\mathbb{R}}x^m\mu (dx)\:\:\: $for all $m \in \mathbb{N}$.
\end{nota}
\begin{rem}
Existence of $\mu$ for $X$ which satisfies $X\sim_{\varphi} \mu$ always holds. The uniqueness of such $\mu$ holds if the moment problem is determinate.
\end{rem}

\begin{df}[Quantum harmonic oscillator] A quantum harmonic oscillator is a triple 
$(\Gamma(\mathbb{C}), a, a^{\ast})$ where $\Gamma(\mathbb{C})$ is a Hilbert space 
$\Gamma(\mathbb{C}):=\oplus^{\infty}_{n=0} \mathbb{C}\Phi _n$ with inner product given by $\langle\Phi_n, \Phi_m\rangle=\delta _{n,m}$, 
and $a, a^{\ast}$ are operators defined as follows:
\[
a\Phi_0 =0, \:\:\:a\Phi_n=\sqrt{n}\Phi_{n-1} (n\geq 1)
\]
\[
a^{\ast}\Phi _n=\sqrt{n+1}\Phi_{n+1} \]
\end{df}

Let $\mathcal{A}$ be the ${\ast}$-algebra generated by $a$, 
and $\varphi_n$ be the state defined as $\varphi(\cdot):=\langle\Phi_n, (\cdot)\Phi_n\rangle$. 
Then $(\mathcal{A}, \varphi_n$ is an algebraic 
probability space. It is well known that 
\[
X:=\frac{1}{\sqrt{2}}(a+a^{\ast})
\]
represents the ``position'' and that 
\[
X\sim _{\varphi_0} \frac{1}{\sqrt{2\pi}}e^{-\frac{1}{2}x^{2}}dx.
\]
That is, in $n=0$ case, the distribution of position is Gaussian. 

On the other hand, the asymptotic behavior of the distributions of position as $n$ tends to infinity is quite nontrivial.

\section{Emergence of the Arcsine law}

\begin{thm}Let $\mu_N$ be a probability distribution on $\mathbb{R}$ such that 
\[
\frac{X}{\sqrt{N}}\sim _{\varphi_N} \mu_N .
\]
Then $\mu_N$ weakly converges to $\mu_{As}$.
\end{thm}
\begin{proof}
We only have to prove moment convergence because it is known that moment convergence implies weak convergence when the moment problem for the limit distribution is determinate.

First we can easily prove that
\[
\varphi_N((\frac{X}{\sqrt{N}})^{2m+1})=\langle\Phi_N,(\frac{a+a^{\ast}}{\sqrt{2N}})^{2m+1}\Phi_N \rangle=0
\]
since $\langle\Phi_N, \Phi_M\rangle=0$ when $N\neq M$. 

To consider the moments of even degrees, we introduce the following notations:
\begin{itemize}
 \item $\Lambda^{2m}:=\{\text{maps from $\{1,2,..., 2m\}$ to $\{1,\ast\}$}\}$, 
 \item $\Lambda^{2m}_{m}:=\{\lambda \in \Lambda^{2m}; |\lambda^{-1}(1)|=|\lambda^{-1}(\ast)|=m\}$.
\end{itemize}

Note that the cardinality $|\Lambda^{2m}_{m}|$ equals to $\binom{2m}{m}$ because the choice of $\lambda$ is equivalent to the choice of $m$ elements which consist the subset $\lambda^{-1}(1)$ from 
$2m$ elements in $\{1,2,..., 2m\}$.

It is clear that for any $\lambda \notin  \Lambda^{2m}_{m}$ 
\[
\langle\Phi_N, a^{\lambda_1}a^{\lambda_2}\cdots a^{\lambda_{2m}}\Phi_N\rangle=0
\]
since $\langle\Phi_N, \Phi_M\rangle=0$ when $N\neq M$.

On the other hand, for any $\lambda \in \Lambda^{2m}_{m}$ the inequality 
\[
N\cdots(N-m+1)\leq \: \langle\Phi_N, a^{\lambda_1}a^{\lambda_2}\cdots a^{\lambda_{2m}}\Phi_N\rangle\: \leq (N+1)\cdots (N+m)
\]
holds when $N$ is sufficiently large, because the minimum is achieved when 
\[
\lambda_i =\left\{
\begin{array}{ll}
1 , &\quad  (1\leq i \leq m)\\
\ast , &\quad  (m+1 \leq i \leq 2m)
\end{array}
\right.
\]
and the maximum is achieved when 
\[
\lambda_i =\left\{
\begin{array}{ll}
\ast , &\quad  (1\leq i \leq m)\\
1 , &\quad  (m+1 \leq i \leq 2m)
\end{array}
\right.
\]
by the definition of $a, a^{\ast}$.

Using the inequality above we have 
\[
\frac{1}{N^m} \langle\Phi_N, a^{\lambda_1}a^{\lambda_2}\cdots a^{\lambda_{2m}}\Phi_N\rangle \:\rightarrow 1 \:\:\:\:(N\rightarrow \infty).
\]
and then 
\[
 \varphi_N((\frac{X}{\sqrt{N}})^{2m})=\langle\Phi_N,(\frac{a+a^{\ast}}{\sqrt{2N}})^{2m}\Phi_N \rangle 
\]
\begin{eqnarray*}
=&\frac{1}{2^m}\sum_{\lambda\in 
\Lambda^{2m}} \frac{1}{N^{m}}\langle\Phi_N, a^{\lambda_1}a^{\lambda_2}\cdots a^{\lambda_{2m}}\Phi_N\rangle & \\
\\
=&\frac{1}{2^m}\sum_{\lambda\in \Lambda^{2m}_{m}} \frac{1}{N^{m}}\langle\Phi_N, a^{\lambda_1}a^{\lambda_2}\cdots a^{\lambda_{2m}}\Phi_N\rangle  & 
\end{eqnarray*}
\[
\rightarrow \frac{1}{2^m}|\Lambda^{2m}_{m}|=\frac{1}{2^m}\binom{2m}{m} \:\:\:\:(N\rightarrow \infty).\\
\]
\end{proof}

\begin{rem}
The theorem above can be extended to the cases for $q$-Fock spaces ($0\leq q\leq 1$), 
which are typical example of ``interacting Fock spaces \cite{A-B}'', if we just replace integer $N$ by 
\[
N_q:=1+q+q^2+\cdots +q^{N-1}.
\]
The proof is quite similar and we omit it here.
\end{rem}

\section{Summary and prospects}
As we have stated, the Arcsine law as the time-averaged distribution for classical harmonic oscillator 
emerges from the distributions for quantum harmonic oscillators. 
This is nothing but a noncommutative probabilistic realization of 
Quantum-Classical Correspondence for harmonic oscillators.
 The ``time averaged'' nature is deeply related to the notion of Bohr's 
``complementarity''for energy and time. Starting from energy eigenstates, 
one cannot obtain the classical harmonic oscillator itself but time averaged distribution of it. 

Mathematically, the result above also shows an important aspect of the Arcsine law 
as the universal distribution for 
many kinds of ``interacting Fock spaces''\cite{A-B} (not only $q$-Fock spaces) which are deeply connected to the theory of orthogonal polynomials 
(This point of view is due to Professor Bo\.{z}ejko). 
The condition for interacting Fock spaces (orthogonal polynomials) 
from which the Arcsine law emerges as the 
high-energy limit distribution should be discovered.

\section*{Acknowledgments} 
The author would like to thank Prof. Marek Bo\.{z}ejko and the referee for interests and comments. 
He is greatly indebted to Prof. Izumi Ojima and Mr. Kazuya Okamura  
for discussions on Quantum-Classical Correspondence.

\end{document}